\documentclass[12pt,a4paper]{article}

\usepackage[utf8]{inputenc}

\usepackage[T1]{fontenc}   
\usepackage{lmodern}       
\usepackage{textcomp}      
  
  \usepackage{bm} 
  
\usepackage{placeins}
\usepackage{amssymb}
\usepackage{amsfonts,amsmath,amsfonts,graphicx,graphics,epsfig,subfig}

\usepackage{dsfont} 

\def\C{\bm{C}}

\DeclareMathOperator*{\argmin}{argmin}%
\newtheorem{theorem}{Theorem}

\newtheorem{rem}{Remark}
\newtheorem{proposition}{Proposition}
\newtheorem{corollary}{Corollary}

  \newtheorem{definition}{Definition}

\newtheorem{prop}{Proposition}
\newtheorem{proof}{Proof}
\newtheorem{thm}{Theorem}
\newtheorem{lem}[thm]{Lemma}

\def \m {\mathbf}
\def \R {\mathbb R}

\def \beg{\begin{eqnarray}}
\def \en{\end{eqnarray}}
\def \be*{\begin{eqnarray*}}
\def \e*{\end{eqnarray*}}
\def \di{\displaystyle}
\def \bit{\begin{itemize}}
\def \eit{\end{itemize}}
\def \E{\mathbb E}
\def \N{\mathbb N}
\def \I{\mathbb I}
\def\R{\mathbb{R}}
\def\L{\mathbb{L}}

\def\C{\bm{C}}

\def \si{\sigma}
\def \al {\alpha}
\def \V {\mathbb{V}}
\def \w{\widehat}
\def \td{\tilde}

\def\argmax{\mathop{\mbox{\sl\em argmax}}}
\renewcommand{\L}{\mathbb{L}}
\renewcommand{\P}{\mathbb{P}}
\renewcommand{\hat}{\widehat}
\renewcommand{\tilde}{\widetilde}
\renewcommand{\d}{\mbox{ \sl\em d}}
\newcommand{\ds}{\displaystyle}
\newcommand{\noi}{\noindent}
\renewcommand{\L}{\mathbb{L}}
\renewcommand{\P}{\mathbb{P}}
\def\argmax{\mathop{\mbox{\sl\em argmax}}}

\title{\huge Lancaster copulas}
\usepackage{authblk}

\author[1]{Angelo Efoévi Koudou}
\author[2]{Yves I. Ngounou Bakam}
\author[3]{Denys Pommeret}

  \affil[1]{Université de Lorraine, CNRS, IECL, F-54000 Nancy, France, E-mail: efoevi.koudou@univ-lorraine.fr}

  \affil[2]{De Vinci Higher Education, De Vinci Research Center, Paris, France, E-mail: yves.ngounou@devinci.fr}

  \affil[3]{I2M, CNRS, Aix Marseille University, Marseille, France,  E-mail: denys.pommeret@univ-amu.fr}

\usepackage[pdftex, colorlinks=true]{hyperref}
\hypersetup{
	bookmarks=true,         
	unicode=false,          
	pdftoolbar=true,        
	pdfmenubar=true,        
	pdffitwindow=false,     
	pdfstartview={FitH},    
	pdftitle={My title},    
	pdfauthor={Author},     
	pdfsubject={Subject},   
	pdfcreator={Creator},   
	pdfproducer={Producer}, 
	pdfkeywords={keywords}, 
	pdfnewwindow=true,      
	colorlinks=true,	    
	linkcolor=blue,        
	citecolor=blue,        
	filecolor=blue,        
	urlcolor=blue          
}

\begin{document}
\date{}
\maketitle

\begin{abstract}

We introduce a new copula class, called Lancaster copulas, built from orthogonal expansions of continuous Lancaster probabilities. We derive infinite-series representations for the copula and its density, study truncation effects, and show in numerical experiments that low-order truncations already provide accurate approximation.\\


\end{abstract}

\newpage
\section{Introduction}\label{sec1_chap5}

Consider a random vector $(X,Y)$ with joint distribution function $F_{X,Y}$ and marginal distribution functions $F_X, F_Y$, that we assume to be continuous.
 According to Sklar \cite{Sklar}, there exists a unique bivariate
 function $C$ such that
 \beg
 F_{X,Y}(x,y) &=& C(F_X(x),F_Y(y)).
 \en
  The function $C$ is called the copula (function) associated with  $(X,Y)$. The copula is a joint  distribution
 function on $[0, 1 ]^2$ with uniform margins and satisfying
$C(u, v) = F_{X,Y}(F_X^{-1} (u), F_Y^{-1} (v))$, where $F^{-1}_X(u) = \inf\{x;  F_X(x) \geq  u\}$,
is the quantile function of $X$.
Assuming that  $F_{X}$ and $F_Y$ are differentiable, we can express the joint density $f_{X,Y}$ of $(X,Y)$ as
 \beg
\label{cop1_chap5}
f_{X,Y}(x,y)  = c\left(F_{X}(x),F_{Y}(y)\right)f_{X}(x)f_Y(y),
\en
where  $f_{X}$ (resp. $f_Y$) is the marginal density of ${X}$ (resp. $Y$) and where
\beg
\label{density_chap5}
c(u,v) & =  & \frac{\partial^{2} C(u,v)}{\partial u  \partial v }
\en
is called the copula density of $(X,Y)$.
Copulas are widely applied in statistics and related fields (e.g.,  \cite{Joe2014, Nelsen2007}), and their use has recently extended to applied domains such as machine learning 
(see \cite{Ren2022, Hernandez2024}). 
%
Several copula families or classes are based on  Sklar's Theorem  and have been derived from certain families of multivariate distributions. 
Various constructions have been proposed, such as, for instance, Archimedean copulas \cite{Hofert2013}, Vines see \cite{Aas2009}, Elliptical copulas \cite{Frahm2003} which are  the distribution functions of componentwise transformed elliptically distributed random vectors, or the Extreme-value copulas  \cite{Gudendorf2010} which are the copulas of random vectors following multivariate extreme value distributions. A review on copulas dependence properties can be found in Ansari and Rockel \cite{Ansari2024}. Extensions of copulas through mixtures or transformations offer a wide range of modeling possibilities (see for instance \cite{Saminger2024, He2024}). Let us also mention Quessy \cite{Quessy2024}, who studies copulas resulting from nonmonotone transformations. And new constructions of copulas appear in recent literature, as in Guzmics and Pflug \cite{Guzmics2020} with multivariate generalizations of the exponential distribution, or in Pfeifer et al.\cite{Pfeifer2019}) with copulas that take into account tail  dependence and asymmetry. 

In this paper we propose a novel copula class called {\it Lancaster copulas}. This class is based on distributions satisfying a bi-orthogonality condition involving orthogonal polynomials with respect to their
marginal distributions. In other words, we consider the copulas of random vectors whose joint distributions belong to the class of Lancaster distributions. 
%
The reader is referred to Lancaster \cite{Lancaster1958,Lancaster1963,Lancaster1975}, the synthesis
in   \cite{Koudou1995,Koudou1996} and the characterization by \cite{Pommeret2004, Koudou2000} for more details.
To recapitulate this class of Lancaster distributions, we follow the standard
notation used in the literature, as given by these authors.

Let $(X,Y)$ be a random vector with  a bivariate distribution $\sigma$ and  margins $\mu$ and $\nu$ defined on  $\mathcal{X}$ and $\mathcal{Y}$, respectively. 
Assume that there exists orthonormal polynomials, $(P_n)_{n\in \N}$  (resp. $(Q_n)_{n\in \N}$), with respect to $\mu$ (resp. $\nu$); that is,    for $n,k\geq 0$
\begin{align*} 
\di \int_{\mathcal{X}} P_n(x)P_k(x)\mu(dx)  =  \di \int_{\mathcal{Y}} Q_n(x)Q_k(x)\nu(dx)  & =  \delta_{n,k},  
\end{align*}
where $\delta_{n,k}=1$ if $n = k$, and 0 otherwise. 
Then the  bivariate distribution $\sigma$ is said to be a Lancaster probability if it satisfies the following bi-orthogonality property 
\begin{align}
\E\left(P_n(X)Q_k(Y)\right)&= \rho_n\delta_{n,k}, 
\label{egaliteLancaster3}
\end{align}
 and the sequence
$\rho_n =\E\left(P_n(X)P_n(Y)\right)$ is called the Lancaster sequence of $\sigma$.  
It is assumed that 

\be*
    {\rm (A1)} &\hspace*{2cm}& \sum_{n \in \N}\rho_{n}^2 < \infty.\hspace*{11cm}  
\e* 
Then the series
$\sum_{n \in \N}\rho_{n}P_{n}Q_{n}$ converges in 
$ L^2(\mu \otimes \nu)$
and we have 
\beg
\label{expansion}
\di \sigma(dx,dy)=\sum_{n \in \N}\rho_{n}P_{n}(x)Q_{n}(y)\mu(dx)\nu(dy). 
\en

Lancaster  distributions are widely used in various
areas of probability and statistics. For instance, they appear in the literature concerning stationary Markov
processes (see \cite{Bussgang1952},
\cite{Barrett1955}, \cite{Brown1958},
\cite{Wong1962}, \cite{Wong1964}) and also in canonical analysis (see
\cite{Dauxois1975}, \cite{Buja1990}).  More recently Cuadras \cite{Cuadras05} shown that  the connection between Lancaster distributions and continuous canonical correlation analysis
(CCA) can be understood through their common spectral decomposition of dependence.
Lancaster  distributions have also been applied in disjunctive kriging within geostatistics 
(see  \cite{Droesbeke2002} and  references 
therein). There are also a few more recent references:
Griffith \cite{Griffiths2009} on reversible Markov
processes whose eigenfunctions are orthogonal
polynomials, Diaconis et al. \cite{Diaconis2008}  where new techniques
are proposed to calculate convergence rates to stationarity  in the context of bivariate Gibbs sampling and
where it is shown that these techniques yield accurate results  
if the target measurements are Lancaster distributions. 
In a related direction, Mena and Palma \cite{Mena2020} developed a framework linking Lancaster probabilities with reversible continuous-time Markov processes, by exploiting orthogonal polynomial eigenfunctions.
In the area of discrete probabilities,
Diaconis and Griffiths \cite{Diaconis2012} provides an interpretation
of Lancaster probabilities of binomial margins using the
generalized Ehrenfest ballot box model and
Griffiths and Spanò \cite{Griffiths2013} for Dirichlet measurements. We can also mention the problem of calculating distance correlation coefficients between random vectors whose joint distributions belong to the class of Lancaster distributions which has been studied in Dueck et al. \cite{Dueck}. 

Our construction of copulas is based on relation (\ref{expansion}). We obtain a new class of copulas that we call {\it Lancaster copulas}. They are expressed as infinite series in the orthogonal bases that characterize the Lancaster distributions. 
We approximate copulas and copula densities by truncation and show in numerical experiments that low-order truncations already provide accurate results. 
A standardized Lancaster copula is defined by transforming the original random variables into uniform random variables.    Through the property of copula invariance, we also obtain a generalization of continuous Lancaster distributions, replacing the sequences of bi-orthogonal polynomials with series of bi-orthogonal functions.
We  analyze tail behavior and show that Lancaster copulas are asymptotically independent (i.e., they have no tail dependence). Expressions for Spearman's and Kendall's coefficients associated with Lancaster copulas are also obtained.

The  paper is organized as follows: 
Section \ref{sec2_chap5} introduces the Lancaster copulas
and their densities, and discusses their truncated representations. Their transformations are studied and a generalization of the Lancaster probabilities is deduced. 
Some examples  from  \cite{Koudou1996, Goffard2017, Dueck}  are presented in Section \ref{simulation} in the context of Lancaster copulas  where we evaluate the accuracy of their approximations. Special attention is given to the construction of Lancaster copulas from the class of quadratic natural exponential families. 
 In Section \ref{sectail} we proceed to the study of the tail dependence. We also give  expressions for  Spearman's rho and  Kendall's tau associated with  Lancaster copulas. 
An extension to the multivariate case is derived  in Section \ref{secmulti} and a discussion concludes the paper. 

\section{Lancaster copulas and their transformations}\label{sec2_chap5}

\subsection{Assumptions and definition}

Let $\di \mu$ and $\di \nu$ be two probability measures on $\R$
such that there exists an open interval $\Theta$ satisfying  $\di \int e^{\theta x}\mu(dx)< \infty$ and
$\di \int e^{\theta y}\nu(dy)< \infty$, for all $\theta \in \Theta$, respectively (this condition implies the existence of moments of any order for $\mu$ and $\nu$ and that the measures are characterized by their moments). 
Assume that  $(X,Y)$ is a  random vector with margins $\mu$ and $\nu$ and  joint Lancaster distribution $\si$ 
satisfying  
decomposition (\ref{expansion}).   
We assume that $\mu$, $\nu$, and $\sigma$ have densities $f_X$, $f_Y$, and $f_{X,Y}$  with respect to the Lebesgue measure on their supports.
We also make the following assumptions:
\vspace{0.2cm}
\begin{itemize}
    \item[(A2)] \, $f_X$ and $f_Y$ are continuous and strictly positive on their supports ${\cal X}$ and ${\cal Y}$. 
\eit    
\vspace{0.2cm}
        As a consequence of (A2), the quantile maps $F_X^{-1}$ and $F_Y^{-1}$ are continuous on $(0,1)$ and differentiable almost everywhere with derivative $(F_X^{-1})^\prime (u)=1/f_X(F_X^{-1}(u))$.

\begin{prop} \label{prop1}
Assume that $(X,Y)$ is a vector with  margins $\mu$ and $\nu$ and joint Lancaster distribution $\sigma$, satisfying assumptions 
(A1)-(A2). 
Then the associated copula $C$ and copula density $c$ have the following expressions:
\beg
\label{cop_dens_lancaster}
\di c(u,v)=\sum_{n \in \N}\rho_{n}P_{n}(F_{X}^{-1}(u))Q_{n}(F_{Y}^{-1}(v)),
\en
\beg
\label{cop_lancaster}
\di C(u,v)=\sum_{n \in \N}\rho_{n}\int_{-\infty}^{F_{X}^{-1}(u)}P_{n}(x)f_X(x)dx  \int_{-\infty}^{F_{Y}^{-1}(v)}Q_{n}(y)f_Y(y)dy.
\en
\end{prop}
\begin{proof} 
From (A2), $f_X$ and $f_Y$ are strictly positive, which implies that the density copula $c$ exists. By construction, $c$ satisfies   
\be* 
f_{X,Y}(x,y) & = & 
f_X(x)f_Y(y) c(F_X(x),F_Y(y))
\e*
that we identify with (\ref{expansion}) to obtain  (\ref{cop_dens_lancaster}).

To prove (\ref {cop_lancaster})  we observe that
\be*
C(u, v)
& = & 
 F_{X,Y}(F_X^{-1} (u), F_Y^{-1} (v))\\
& = &
P(X\leq F_X^{-1} (u), Y \leq F_Y^{-1} (v))\\
& = &
\int_{\R ^2} 1_{(-\infty, F_X^{-1} (u))}(x)
1_{(-\infty, F_Y^{-1} (v))}(y) \, \sigma(dx, dy)\\
& = &
\int_{\R ^2} 1_{(-\infty, F_X^{-1} (u))}(x)
1_{(-\infty, F_Y^{-1} (v))}(y) \, (\sum_{n\in \N}\rho_n P_n(x)Q_n(y)) \mu(dx) \nu(dy) \\
& = &
\langle    g, \lim_{N\to \infty} h_N   \rangle_{L^2(\mu \otimes \nu)},
\e*
where $\langle  ,  \rangle_{L^2(\mu \otimes \nu)}$ is the inner product in the Hilbert space $L^2(\mu \otimes \nu)$ and 
$$g(x,y)= 1_{(-\infty, F_X^{-1} (u))}(x) \, 1_{(-\infty, F_Y^{-1} (v))}(y)$$
and, for $N \in \N$ and $(x,y) \in \R ^2$,
$$h_N(x,y)=\sum_{n=0}^N \rho_n P_n(x)Q_n(y).$$
By continuity of the inner product in $L^2(\mu \otimes \nu)$ we have 
$$C(u, v)=  \lim_{N\to \infty} \langle    g,  h_N   \rangle_{L^2(\mu \otimes \nu)}$$
 and we get (\ref {cop_lancaster}).
\end{proof}

\begin{definition} 
We call a function $C$ defined on $[0,1]^2$ a Lancaster copula if it is the copula of a random vector $(X,Y)$ whose  joint  distribution $\sigma(dx,dy)$ is a Lancaster distribution with margins $\mu$ and $\nu$ satisfying assumptions (A1)-(A2).
\end{definition}


It is worth pointing out that (\ref{cop_dens_lancaster}) has similarities  with a diagonal spectral decomposition as studied in Cuadras \cite{Cuadras05} where 
the coefficients $\rho_n$ play the role
of canonical correlations, measuring dependence along orthogonal modes, while truncation
corresponds to restricting the analysis to a finite number of dominant canonical components. 

Other related references include Longla and \cite{Longla24} and Muia and Longla \cite{Muia25} who studied  symmetric copula densities having similar form  to  (\ref{cop_dens_lancaster}). But such a symmetry in $u$ an d$v$ necessitates that $X$ and $Y$ have the same distribution.

 
\begin{rem}
Discrete Lancaster distributions do exist (see, for instance, Koudou \cite{Koudou1996}), with a representation similar to (\ref{expansion}). 
%
However, no density copula is associated with such a representation  and we therefore restrict our analysis in the present paper to the continuous case.
\end{rem}

\subsection{Standardized Lancaster copula}
Let $(X,Y)$ be a  random vector with Lancaster copula $C$. 
It is well known that under  increasing measurable functions $G$ and $H$, $(X,Y)$ and $(G(X),H(Y))$ share the same copula. 
In particular, if we consider $(\td X, \td Y) := (F_X(X),F_Y(Y))$, with joint 
distribution function  $\td F := F_{\td X,\td Y}$,  we know that 
$C(u,v) =  \td F(u,v)$. Since the margins of $(\td X, \td Y)$ are uniform it follows that 
\beg
\label{cop_dens_stand}
\di c(u,v)=1+ \sum_{n \geq 1}\sum_{ k \geq 1}\theta_{n,k}L_{n}(u)L_{k}(v), 
\en
where $\theta_{n,k}=\E(L_n(\td X)L_k(\td Y))$, with $(L_n)_{n\in\N}$ a basis of shifted Legendre polynomials orthonormal with respect to the uniform measure on $[0,1]$, with $L_0 \equiv 1$ and $L_1=\sqrt{3}(2x-1)$. We call $\theta_{n,k}$ the {\it standardized coefficients}.  Since the margins are uniform, we have
$\theta_{0,0}=1$,
 and   
$\theta_{n,0}=\theta_{0,k}=0$,
 for $n,k\ge 1$. 
From (\ref{cop_dens_stand}) we deduce 
\beg
\label{cop_stand}
\di C(u,v)=uv + \sum_{n \geq 1}\sum_{k \geq 1}\theta_{n,k}\displaystyle \int_{0}^{u}L_{n}(x)\mu(dx)\int_{0}^{v}L_{k}(y)\nu(dy).
\en
Representations (\ref{cop_dens_stand})-(\ref{cop_stand}) and   (\ref{cop_dens_lancaster})-(\ref{cop_lancaster}) are equivalent. 
We call  (\ref{cop_dens_stand})-(\ref{cop_stand}) the {\it standardized representations}, that is, the representation associated with  uniform margins in the Legendre basis.  
While margin transformation does not preserve bi-orthogonality (see the related work of \cite{Pommeret2005}), it does preserve  the Lancaster copula and thus allows Lancaster's class to be generalized.

\subsection{Generalized  Lancaster distributions}
By the copula invariance property, if $(X,Y)$ has a Lancaster copula  $C$ satisfying   (\ref{cop_dens_lancaster}) and (\ref{cop_lancaster})  and if  $(\td X, \td Y) := (G(X),H(Y))$, with $G,H$  increasing measurable functions, we have 
\be*
F_{\td X,\td Y}(x,y) & =& C(F_{\td X}(x), F_{\td Y}(y)).
\e*
Differentiating this equality we obtain 
\beg \label{dens_lancasterbis}
\di f_{\td X,\td Y}(x,y)=f_{\td X}(x)f_{\td Y}(y) \sum_{n\in \N}\rho_n P_{n}(F_{X}^{-1}(F_{\td X}(x)))Q_{n}(F_{Y}^{-1}(F_{\td Y}(y))),
\en 
with $\rho_n=\E(P_n(X)Q_n(Y)$. 
Define 
\beg 
\label{newpol}
\td P_n(x) = P_{n}(F_{X}^{-1}(F_{\td X}(x))) & {\rm and } & \td Q_n(y) = Q_{n}(F_{Y}^{-1}(F_{\td Y}(y))), 
\en
we obtain two sequences of bi-orthonormal functions. More precisely we have 
\be*
\E\left(\td P_n(\td X) \td P_k(\td X) \right)
=\E\left(\td Q_n(\td Y) \td Q_k(\td Y) \right)= \delta_{n,k},  
\ \E\left(\td P_n(\td X)\td Q_k(\td Y)\right)= \rho_n\delta_{n,k}.
\label{generalized3}
\e*

Consequently, applying $G,H$ on the margins of  a Lancaster distribution $\si$, we obtain a new distribution $\td \si$ with margins $\td \mu$ and $\td \nu$, satisfying 
\beg
\label{expansion2}
\di \td \sigma(dx,dy)&=& \di \sum_{n \in \N}\rho_{n}\td P_{n}(x)\td Q_{n}(y)\td \mu(dx) \td \nu(dy), 
\en 
where $\td P_n$ and $\td Q_n$ are the bi-orthonormal functions defined by (\ref{newpol}), with $\rho_n  = \E(\td P_n(\td X)\td Q_n(\td Y))$. 
We  call any joint distribution $\td \si$ satisfying (\ref{expansion2}) a {\it generalized Lancaster distribution}.  
  
Ultimately, Lancaster copulas not only enable us  understand the dependency between the components of a Lancaster distribution, but also let  us  understand the dependence structure of any associated  generalized Lancaster distribution. 

\subsection{Truncation} 
Lancaster copula constructions are based on infinite series, which in practice are truncated  to obtain approximations. 

For any  positive  integer $N$ we  define the following $N$-th order approximations 
\be*
\di c^{[N]}(u,v)  = \sum_{n=0}^{N} \rho_{n}P_{n}(F_{X}^{-1}(u))Q_{n}(F_{Y}^{-1}(v))
,\quad (u,v)\in (0,1)^2,
\e*
and
\be*
\di C^{[N]}(u,v)  =\sum_{n=0}^{N}\rho_{n}\int_{-\infty}^{F_{X}^{-1}(u)}P_{n}(x)\mu(dx)\int_{-\infty}^{F_{Y}^{-1}(v)}Q_{n}(y)\nu(dy) ,\quad (u,v)\in (0,1)^2.
\e*

\begin{rem}
As shown in Section \ref{simulation}, low-order truncations already provide good approximations. In estimation, choosing $N$ too large requires estimating high-order moments and may increase variability, whereas choosing $N$ too small may lead to an overly coarse approximation of the target copula.
In Section \ref{discussion} we discuss a data-driven procedure based on the Least-Squares Cross-Validation to automatically select $N$. 
\end{rem}

\begin{rem}
We can adapt the truncation in the standardized case as follows:  we consider two  positive  integers $N_1,N_2$ and  we  define the $(N_1,N_2)$-th order standardized approximation by 
\be*
\di c^{[N_1,N_2]}(u,v)  = 1+ \sum_{n=1}^{N_1} \sum_{k=1}^{N_2}  \theta_{n,k}L_{n}(u)L_{k}(v)
,\quad (u,v)\in (0,1)^2,
\e*
and
\be*
\di C^{[N_1,N_2]}(u,v)  = uv + \sum_{n=1}^{N_1}\sum_{k=1}^{N_2} \theta_{n,k}\int_{0}^{u}L_{n}(s)ds\int_{0}^{v}L_{n}(t)dt ,\quad (u,v)\in (0,1)^2.
\e*
\end{rem}

\subsection{Positivity}
The infinite expansions associated with the Lancaster copulas are nonnegative by construction since they originate from a copula. 
However,  the positivity of the truncated series  is not guaranteed. 
 This  point is essential if the series are used as estimators of the copulas and their densities (see the discussion in Section \ref{discussion}).   
 We propose two partial solutions here. 
 \bit
 \item 
 A first basic  solution is to replace the truncated series $c^{[N]}(u,v)$ by $\max(c^{[N]}(u,v),0)$.  
 In our numerical examples we did not observe violations of nonnegativity; however, this correction may be useful in general estimation settings. 
 In practice, this means that at certain bivariate points, the density will be zero. These areas are undoubtedly areas where the density is low, and we can expect that by increasing the number of observations and/or the degree of truncation, the estimates will all be positive.  
 
 Note that applying this maximum operator means the resulting function may no
longer integrate to exactly 1, and thus a  normalization constant would be required if the
user strictly needs a valid probability density function in practice. 
 \item 
 A second solution is to find a sufficient condition of positivity (since the necessary and sufficient condition is complex and exceeds the scope of this paper). We have the following result: 
 \begin{lem}
\label{lem:pos_general}
Let $\mu$ and $\nu$ be the marginal distributions of $X$ and $Y$, with continuous cdfs $F_X,F_Y$.
Let $(P_n)_{n\ge 0}$ and $(Q_n)_{n\ge 0}$ be orthonormal polynomial bases in $L^2(\mu)$ and
$L^2(\nu)$, with $P_0\equiv Q_0\equiv 1$.
For $N\ge 1$, consider  the truncated Lancaster copula density approximation
\[
c^{[N]}(u,v)\;=\;1+\sum_{n=1}^N \rho_n\,P_n(F_X^{-1}(u))\,Q_n(F_Y^{-1}(v)),
\qquad (u,v)\in(0,1)^2.
\]
If
\begin{equation}\label{cond1}
\sum_{n=1}^N |\rho_n|\,\|P_n\|_\infty \, \|Q_n\|_\infty \;\le\; 1,
\end{equation}
then $c^{[N]}(u,v)\ge 0$ for all $(u,v)\in(0,1)^2$.
\end{lem}
\begin{proof}  ~Fix $(u,v)\in[0,1]^2$. By the triangle inequality,
\[
c^{[N]}(u,v)
\;=\;1+\sum_{n=1}^N \rho_n \,P_n(F_X^{-1}(u))\,Q_n(F_Y^{-1}(v))
\;\ge\; 1-\sum_{n=1}^N |\rho_n|\,|P_n(F_X^{-1}(u))|\,|Q_n(F_Y^{-1}(v))|.
\]
Since $|P_n(F_X^{-1}(u))|\le \|P_n\|_\infty$ and $|Q_n(F_Y^{-1}(v))|\le \|Q_n\|_\infty$, we obtain
\[
c^{[N]}(u,v)\;\ge\;1-\sum_{n=1}^N |\rho_n|\,\|P_n\|_\infty \, \|Q_n\|_\infty.
\]
Therefore, condition (\ref{cond1}) implies $c^{[N]}(u,v)\ge 0$ for all $(u,v)\in(0,1)^2$.
 \end{proof} 
 Restricting our study to the standardized Lancaster copulas 
we have the following result: 
\begin{corollary}
Let $(L_n)_{n\ge 0}$ be the orthonormal shifted Legendre polynomials on $[0,1]$. For $N\ge 1$, consider the $(N_1,N_2)$th truncated standardized density copula 
\[
c^{[N_1,N_2]}(u,v)\;=\;1+\sum_{n=1}^{N_1}\sum_{k=1}^{N_2} \theta_{n,k}\,L_n(u)L_k(v),\qquad (u,v)\in(0,1)^2.
\]
If
\beg
\label{cond2}
\sum_{n=1}^{N_1} \sum_{k=1}^{N_2}  |\theta_{n,k}|\,\sqrt{2n+1}\sqrt{2k+1} & \le & 1,
\en
then $c^{[N_1,N_2]}(u,v)\ge 0$ for all $(u,v)\in(0,1)^2$. 
\end{corollary}
\begin{proof} ~The proof comes from the bound
$\|L_n\|_\infty^2\le 2n+1$ valid for the orthonormal shifted Legendre basis (see \cite{Szego}) combined with proof of  Lemma \ref{lem:pos_general}. 
\end{proof}
\eit


\begin{rem}
Condition (\ref{cond1}) requires boundedness of
$u\mapsto P_n(F_X^{-1}(u))$ and $v\mapsto Q_n(F_Y^{-1}(v))$ on $[0,1]$.
This holds automatically when the supports of $\mu$ and $\nu$ are compact since $P_n,Q_n$
are continuous. For unbounded supports (e.g. Gaussian margins with Hermite polynomials),
$\|\cdot\|_\infty$ may be infinite.

\end{rem}

\section{Examples of bivariate Lancaster Copulas}\label{simulation}

\subsection{Downton exponential bivariate Lancaster Copula}

Let $\mu$ and $\nu$ denote two exponential distributions with parameters $\lambda_1>0$ and $\lambda_2>0 $, respectively. 
For $\rho >0$, let 
$(X,Y)$ be a  random vector  following the Downton Bivariate Exponential distribution, denoted $\text{DBVE}(\lambda_1, \lambda_2, \rho)$,  introduced in Downton \cite{Downton1970} and further studied in Goffard et al. \cite{Goffard2017}. This distribution is a Lancaster probability with margins $\mu$ and $\nu$, with $\rho_n=\rho^n$. It is commonly used to model the joint lifetimes of two components in reliability analysis. These components are assumed to fail after a random number of shocks, occurring at exponentially distributed intervals. The Lancaster copula density associated with $\text{DBVE}(\mu_1, \mu_2, \rho)$ is given by (\ref{cop_dens_lancaster}) 
and its Lancaster copula function is given by
(\ref{cop_lancaster}) for $(u,v)\in [0,1]^2$, 
where $F_{X}$ (resp. $F_Y$)  is the cdf of $X$ (resp. $Y$).
Here, $(P_n)_{n\in \N}$ and $(Q_n)_{n\in \N}$ form complete orthonormal bases with respect to the measures   $\mu$ and $\nu$, respectively, and are given by the generalized Laguerre polynomials.

We note that in expressions (\ref{cop_dens_lancaster}) and (\ref{cop_lancaster}), although the sums are theoretically infinite, they are numerically stable beyond a certain order. Specifically, we observe that the value of the sum remains numerically unchanged whether we consider a truncation order of 5, 100, or even 300. This stability is illustrated by the density graphs and the contour lines shown in Figure \ref{Fig:Downtown}. The copula density exhibits clear positive dependence, with larger values when $u$ and $v$ are simultaneously high (in particular near the upper--right corner and, more generally, along the main diagonal).
For the very low truncation order $N=1$, the approximation may display truncation artifacts (e.g., less smooth or slightly distorted contour lines).
From $N=6$ onward, the overall shape and the contour structure are already very close to those obtained with $N=100$, indicating that a moderate truncation order can provide an accurate approximation in this example.

\begin{figure}[!htbp]
\centering
\vspace*{-0.5cm}
\hspace*{-2cm}
\includegraphics[scale=0.2]{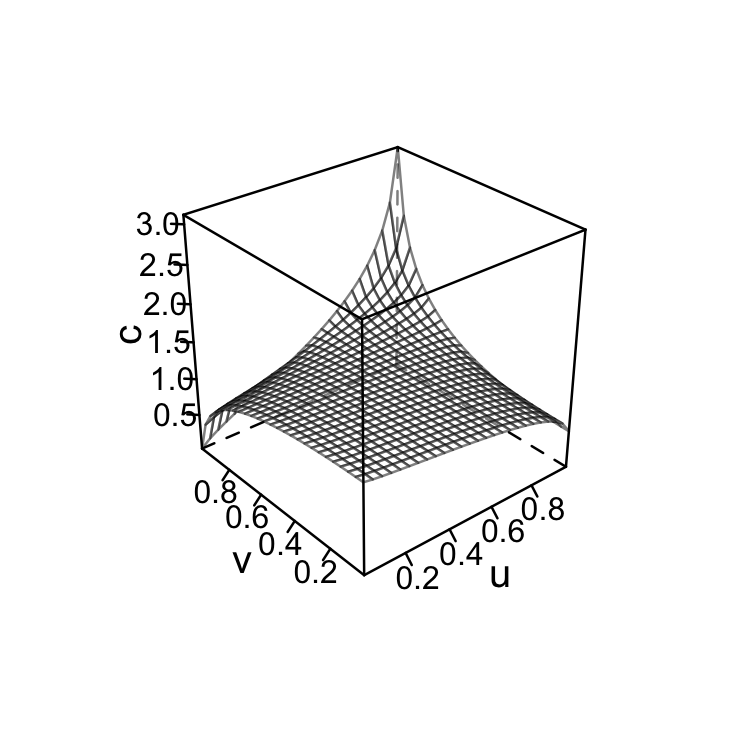}%
\hspace*{-1.2cm}
\includegraphics[scale=0.19]{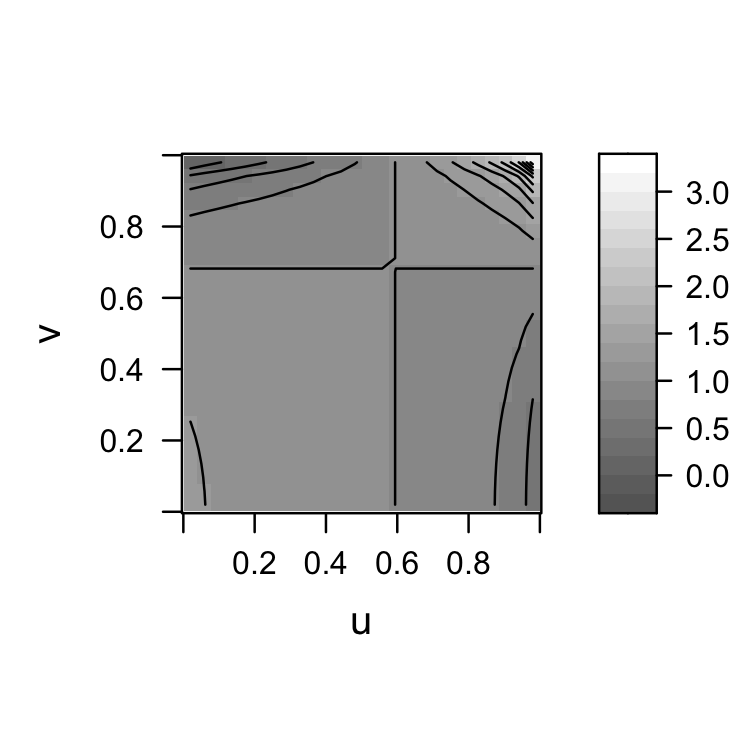}
\\[-0.cm]
\vspace*{-1.1cm}
\hspace*{-2cm}
\includegraphics[scale=0.2]{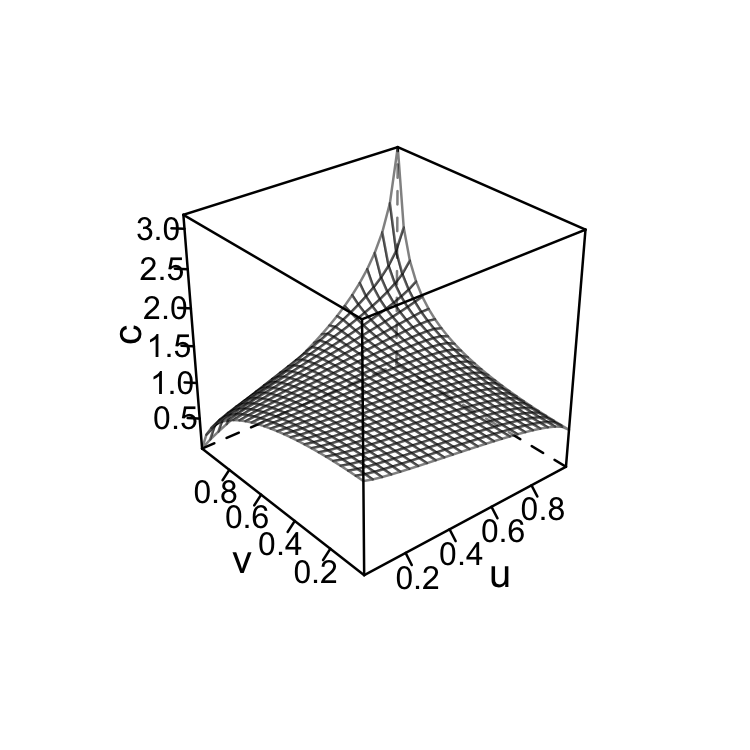}%
\hspace*{-1.2cm}
\includegraphics[scale=0.19]{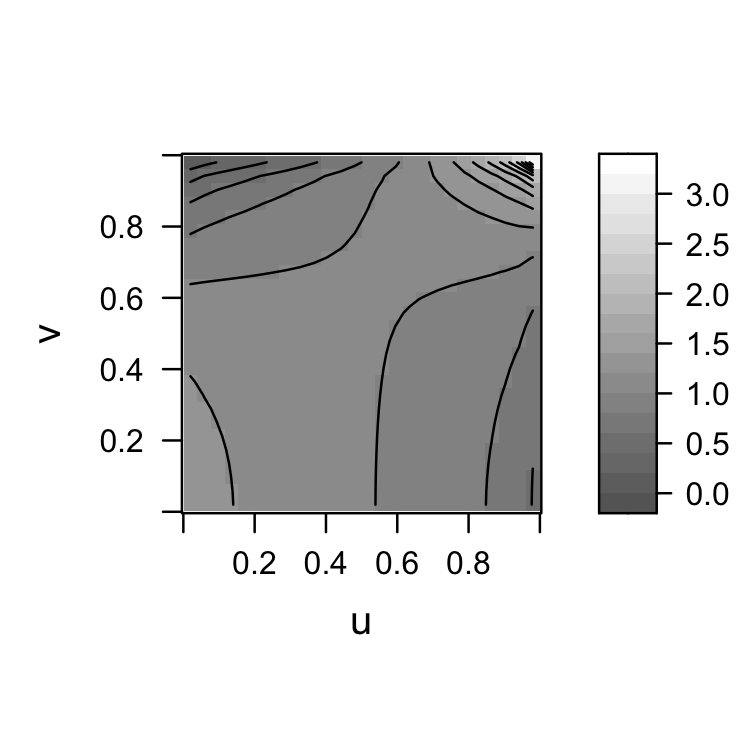}
\\[-0.cm]
\vspace*{-1.1cm}
\hspace*{-2cm}
\includegraphics[scale=0.2]{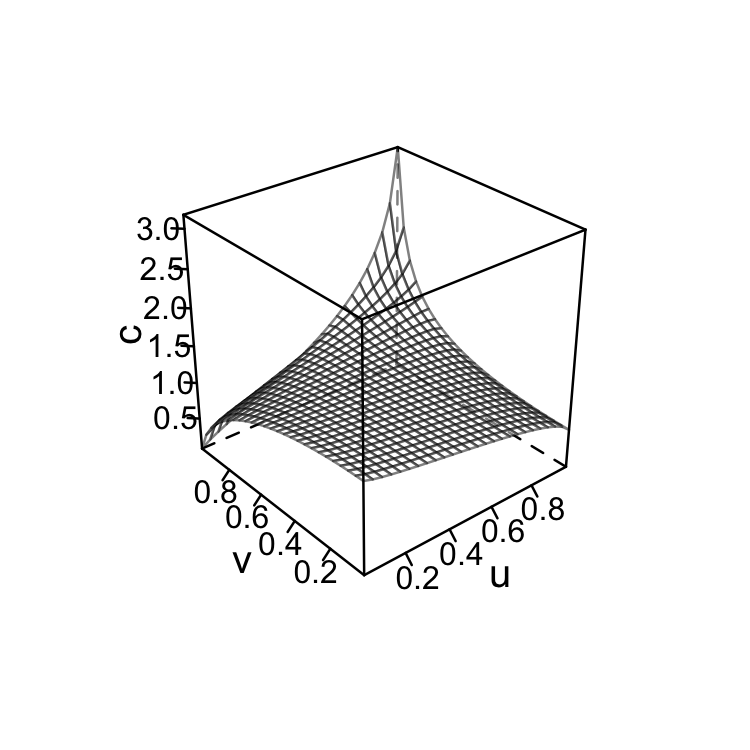}%
\hspace*{-1.2cm}
\includegraphics[scale=0.19]{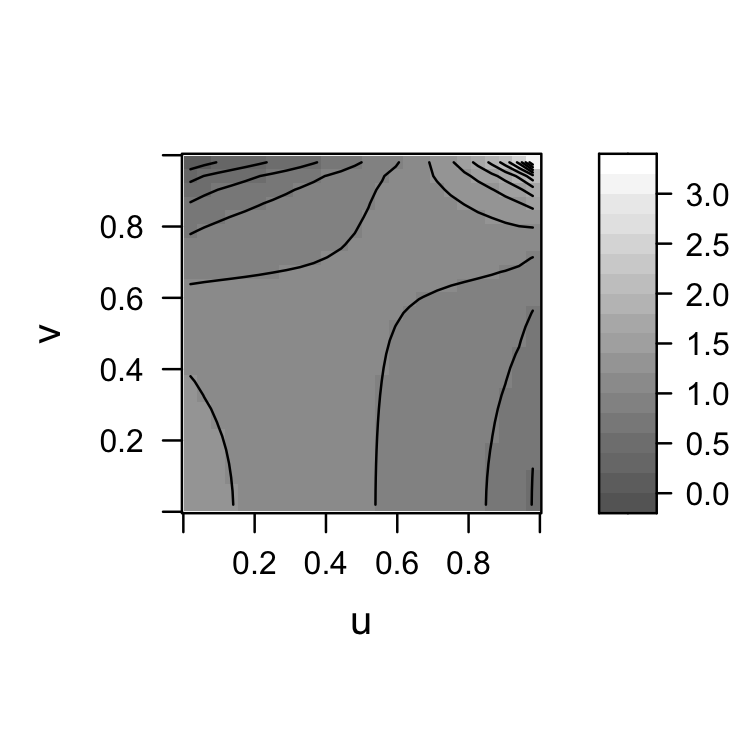}%
\\[-0.9cm]
\caption[]{Density graphs and contour lines  of $\text{DBVE}(\frac{1}{2}, 2, \frac{1}{4})$. Truncation order:  $1$ (top),   $6$ (middle),   $100$ (bottom).}
\label{Fig:Downtown}
\end{figure}

%
\subsection{Bivariate normal Lancaster copula}

Let $\di (X, Y )\sim \mathcal{N}\left(0,\Sigma\right)$ be  a bivariate centered normal vector with  covariance matrix 
$$\Sigma=\left(
\begin{array}{rrr}
	1& r \\
	r & 1 
\end{array}
\right), \ \ \ \ \ \ \ |r|<1.
$$ 
Writing $f_{X,Y}$ the  joint density of $(X,Y)$, and $f_{X},f_{Y}$ the marginal densities,
\cite{Sarmanov1967} considers the  expansion 
\beg\label{exanorm_chap5}
f_{X,Y}(x,y)=f_{X}(x)f_{Y}(y)\sum_{n=0}^{\infty}{\rho_n}H_n(x)H_n(y)
,\quad (x,y)\in \R^2
\en
where $(H_n)_{n\in\N}$ are  Hermite polynomials $\mathcal{N}(0,1)$-orthonormal, 
and where  
$
\rho_n=
    r^n$  for all $n \geq 0$. 
%
Expression (\ref{exanorm_chap5}) shows that $(X,Y)$ has a Lancaster distribution. 
%
%

Therefore, the Lancaster copula density associated with  $(X,Y)$ is given by
\be*
\di c(u,v)=\sum_{n=0}^{\infty}{\rho_n}H_n\left(F_{X}^{-1}(u)\right)H_n\left(F_{Y}^{-1}(v)\right),
 \quad (u,v)\in (0,1)^2
\e*
and its Lancaster copula function is as follows
\be*
\di C(u,v)=\sum_{n \in \N}{\rho_n}\int_{-\infty}^{F_{X}^{-1}(u)}H_{n}(x)\mu(dx)\int_{-\infty}^{F_{Y}^{-1}(v)}H_{n}(y)\nu(dy).
\e*
Figure \ref{Fig:gauss} represents density graphs and contour lines for the bivariate normal Lancaster copula with a truncation order 1, 5, and 100. 
As expected for a Gaussian copula with $r=0.5$, the density is concentrated around the diagonal $u \approx v$, reflecting symmetric positive dependence.
At the lowest order $N=1$, the truncated series may yield a crude representation and can even violate nonnegativity locally when interpreted as a density, which highlights the limitations of overly small truncation orders.
Starting from $N=5$, the contour lines become smooth and the approximation is visually almost indistinguishable from the high-order truncation $N=100$. 
It appears that the numerical values no longer seem to vary  from  order 5. 

\begin{figure}[!htbp]
\centering
\vspace*{-0.5cm}
\hspace*{-2cm}
\includegraphics[scale=0.2]{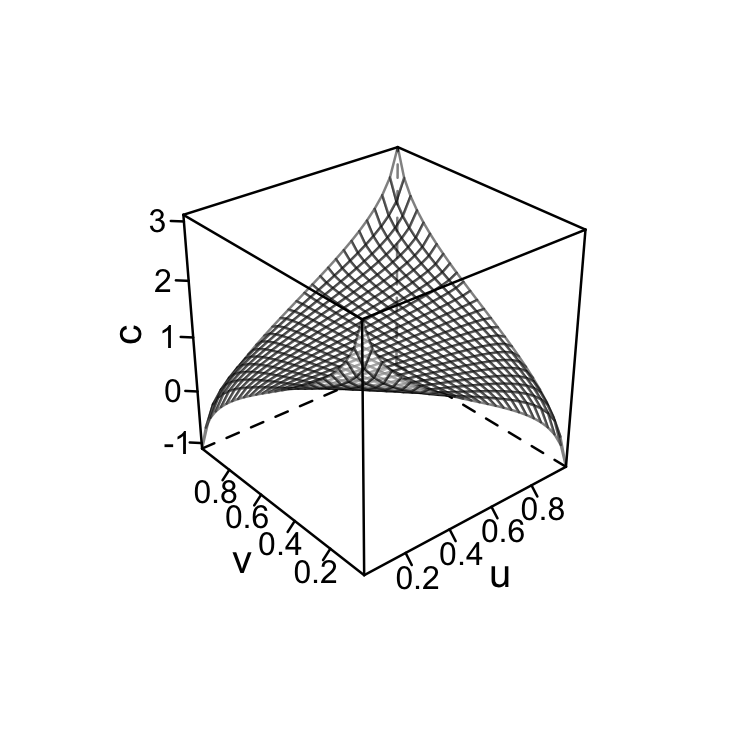}%
\hspace*{-1.2cm}
\includegraphics[scale=0.19]{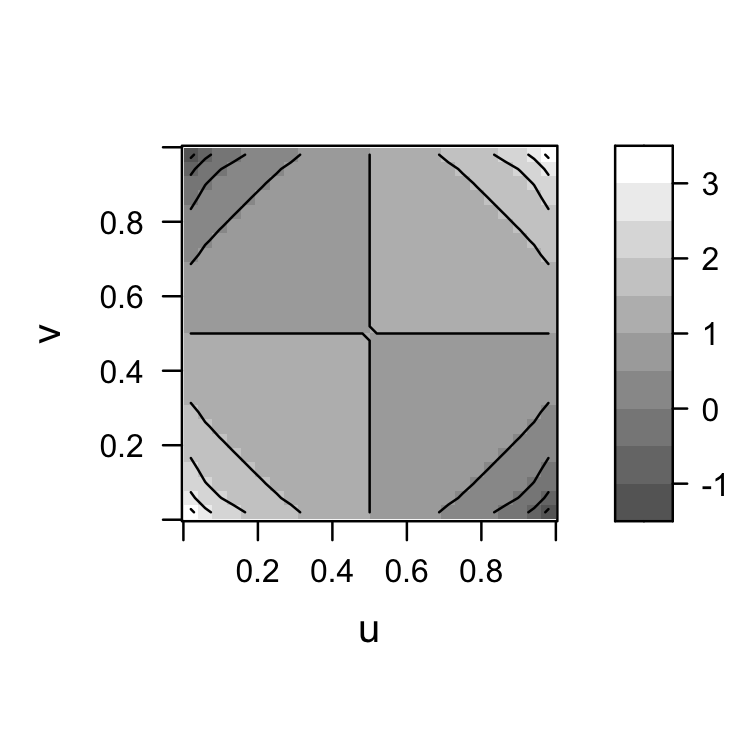}
\\[-0.cm]
\vspace*{-1.1cm}
\hspace*{-2cm}
\includegraphics[scale=0.2]{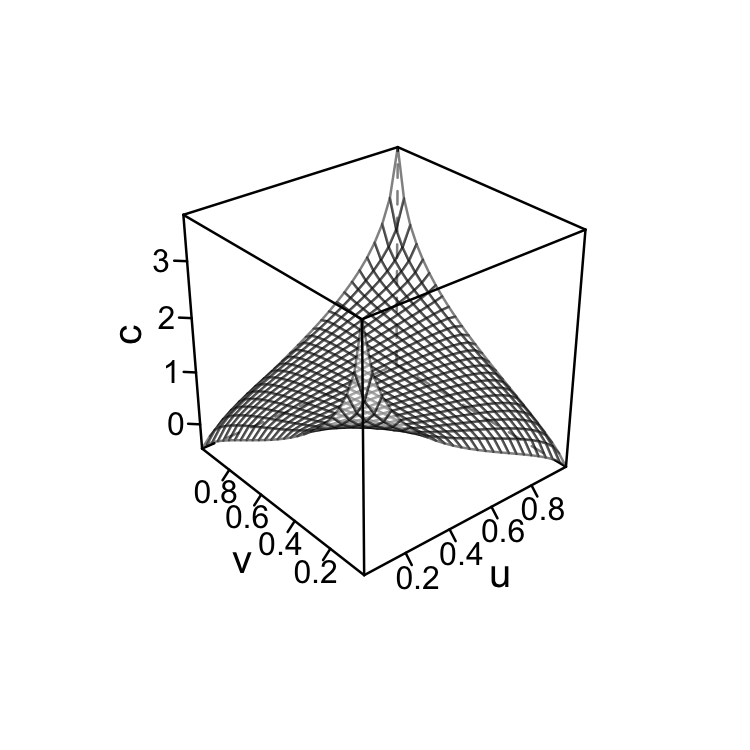}%
\hspace*{-1.2cm}
\includegraphics[scale=0.19]{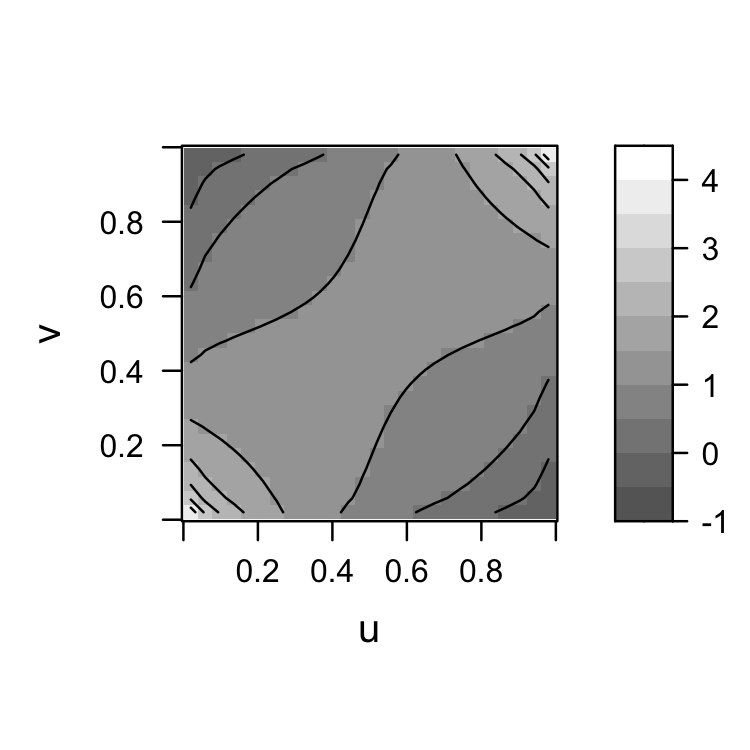}
\\[-0.cm]
\vspace*{-1.1cm}
\hspace*{-2cm}
\includegraphics[scale=0.2]{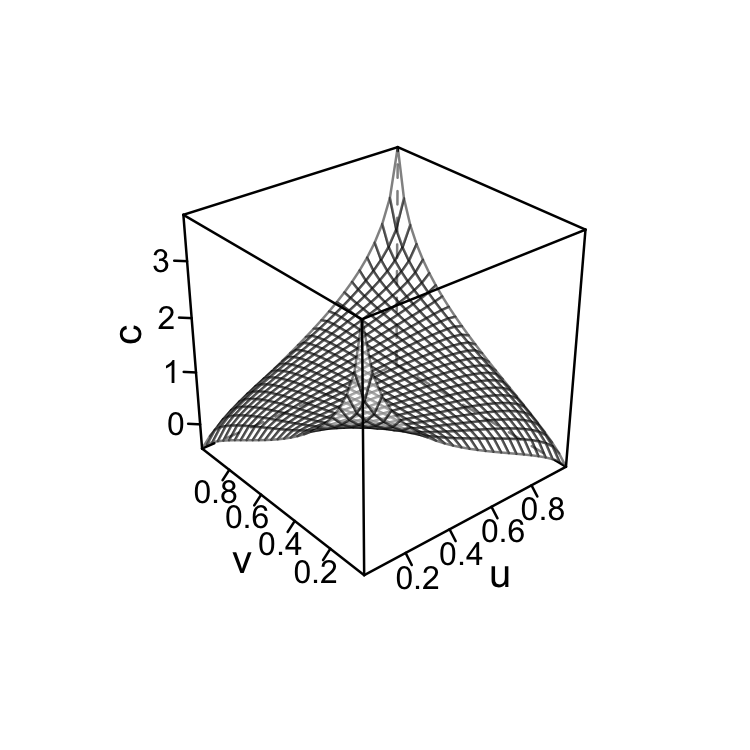}%
\hspace*{-1.2cm}
\includegraphics[scale=0.19]{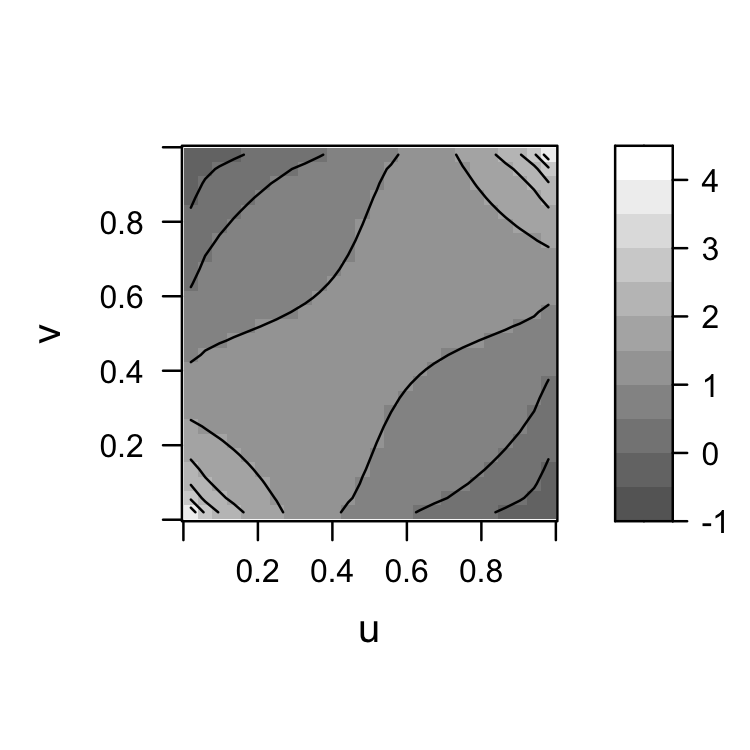}%
\\[-0.9cm]
\caption[]{Density graphs and contour lines of Gaussian Lancaster copula density with parameter $r=0.5$. Truncation order: $1$ (in top),  $5$ (in middle),  $100$ (in bottom).}
\label{Fig:gauss}
\end{figure}

\subsection{Bivariate gamma Lancaster copula}

The Lancaster expansion for a bivariate gamma vector $(X,Y)$,  derived by \cite{Sarmanov1970a,Sarmanov1970b}, and studied in Dueck et al. \cite{Dueck}, can be stated as follows (see Kotz et al. \cite{Kotz2004}) for $x,y \geq 0$: 
\beg\label{gamma_lanc}
\di f_{X,Y}=f_{X}(x)f_{Y}(y)\sum_{n=0}^{\infty}a_n {\mathcal L}_{n}^{(\alpha-1)}(x){\mathcal L}_{n}^{(\beta-1)}(y), 
\en 
where  $\alpha\geq \beta>0$, $a_n =\left(\frac{(\beta)_n}{(\alpha)_n}\right)^{1/2}\lambda^n$, with $\lambda \in (0,1)$,   and  $({{\mathcal L}}_{n}^{(\alpha)})_{n \in \N}$ (resp. $({{\mathcal L}}_{n}^{(\beta)})_{n \in \N}$) are  the Laguerre polynomials  $f_X$ (reps. $f_Y$)-orthonormal. 
%
%
%
%
%
The corresponding marginal density functions are univariate gamma with  parameters $\alpha$ and $\beta$.
%
%
\cite{Kotz2004} proved that if $\alpha=\beta$ 
then the density function (\ref{gamma_lanc}) reduces to the Kibble-Moran bivariate gamma density function with
$Corr(X,Y)=\lambda$ and (\ref{gamma_lanc}) represents the Lancaster expansion for $(X, Y)$. 
Therefore, the Lancaster copula density associated with  the bivariate gamma distribution  $(X,Y)$ is
\be*
\di c(u,v)=\sum_{n=0}^{\infty}\lambda^n{\mathcal L}_{n}^{(\alpha)}\left(F_{X}^{-1}(u)\right){\mathcal L}_{n}^{(\alpha)}\left(F_{Y}^{-1}(v)\right)
,\quad (u,v)\in (0,1)^2, \quad \alpha>1,
\e*
and its Lancaster copula function is given by
\be*
\di C(u,v)=\sum_{n \in \N}\lambda^n\int_{-\infty}^{F_{X}^{-1}(u)}{\mathcal L}_{n}^{(\alpha)}(x)\mu(dx)\int_{-\infty}^{F_{Y}^{-1}(v)}{\mathcal L}_{n}^{(\alpha)}(y)\m(dy),
\e*
where $\mu$ is the gamma measure with density proportional to $x^{\alpha-1}\exp(-x)$. 
Figure \ref{Fig:gamma} shows density graphs and contour lines for the bivariate gamma Lancaster copula  with parameters $\lambda=0.2$ and $\alpha=\beta = 3$, with truncation orders  2, 6, and 100. There is no noticeable difference between the numerical results obtained with these three orders, which once again demonstrates the rapid convergence and quality of the approximations.  Here the dependence is moderate: the density remains close to the independence benchmark (values near $1$) with a mild increase along the diagonal. 
Moreover, the relatively diffuse contours indicate that the density does not concentrate sharply in the extreme corners.

\begin{figure}[!htbp]
\centering
\vspace*{-0.5cm}
\hspace*{-2cm}
\includegraphics[scale=0.2]{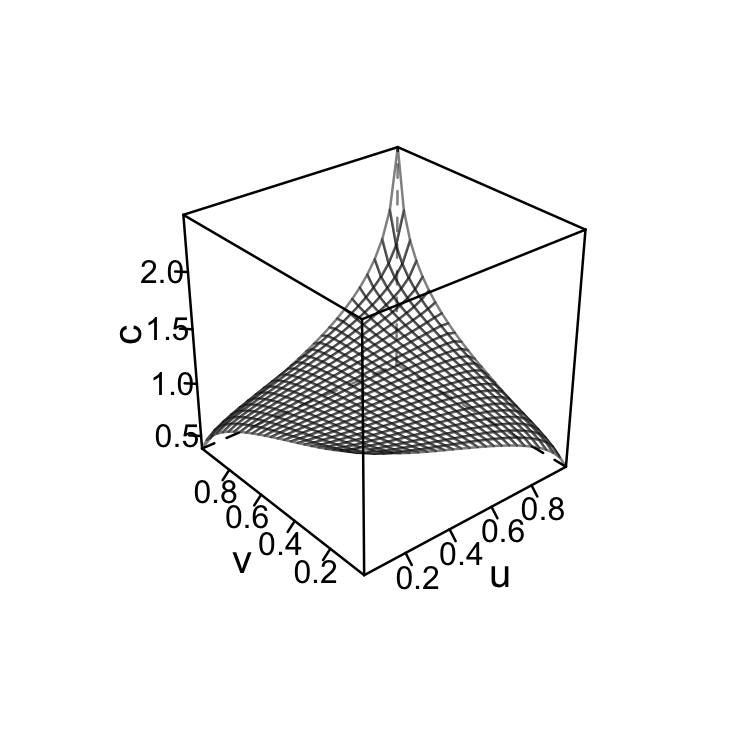}%
\hspace*{-1.2cm}
\includegraphics[scale=0.19]{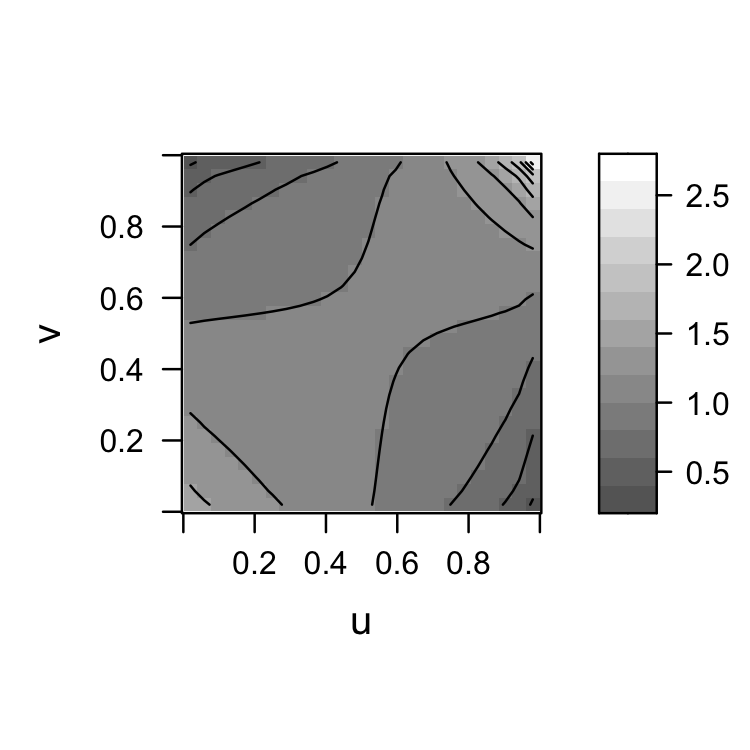}
\\[-0.cm]
\vspace*{-1.1cm}
\hspace*{-2cm}
\includegraphics[scale=0.2]{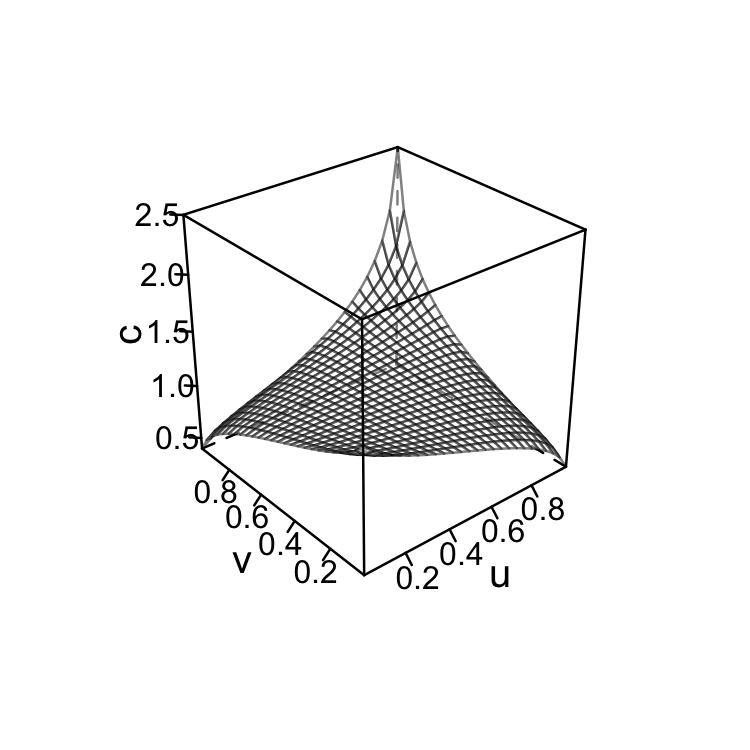}%
\hspace*{-1.2cm}
\includegraphics[scale=0.19]{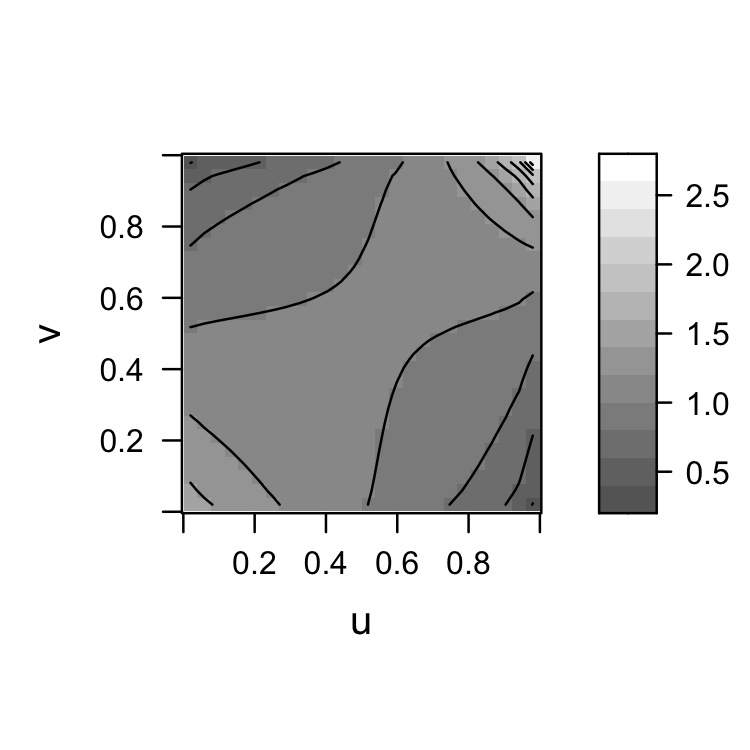}
\\[-0.cm]
\vspace*{-1.1cm}
\hspace*{-2cm}
\includegraphics[scale=0.2]{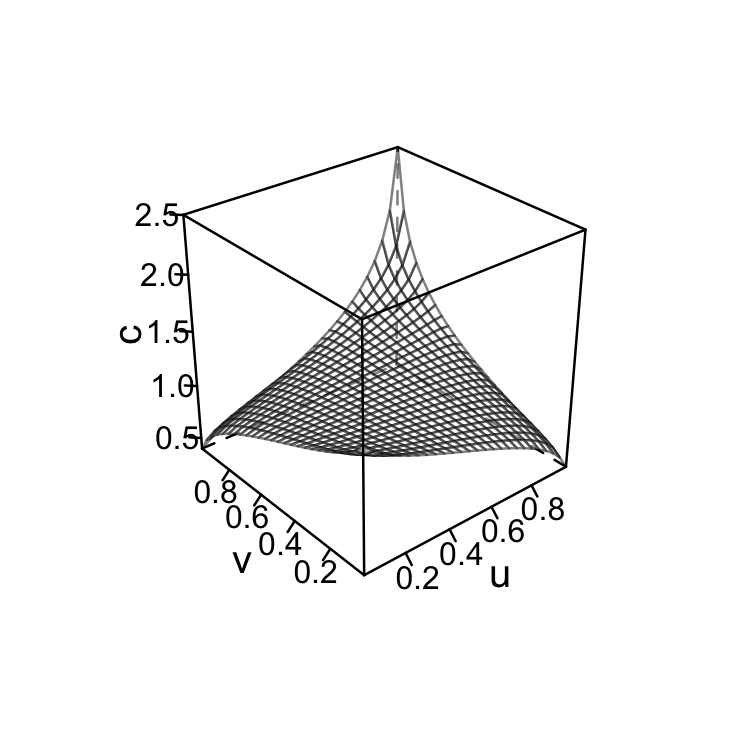}%
\hspace*{-1.2cm}
\includegraphics[scale=0.19]{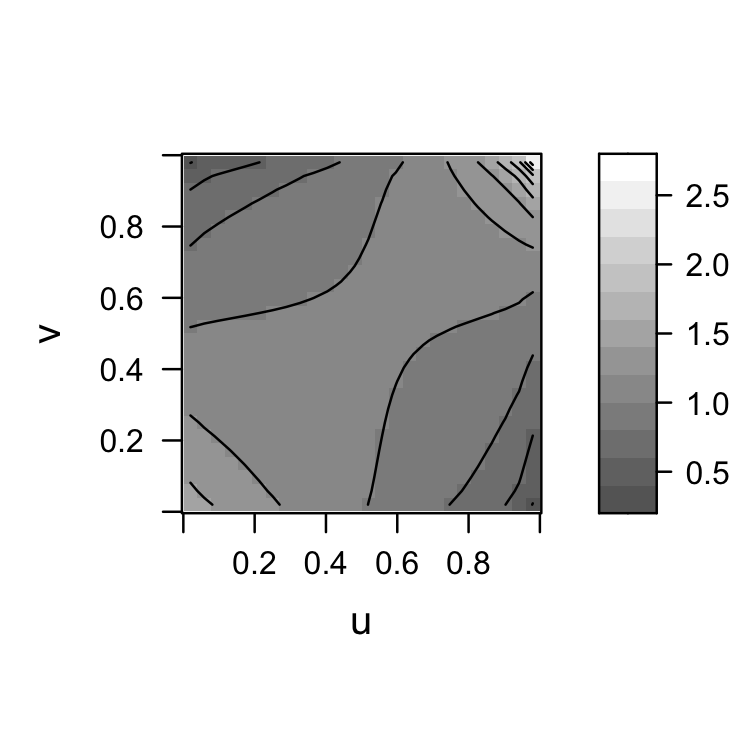}%
\\[-0.9cm]
\caption[]{Density graphs and contour lines of gamma Lancaster copula density with parameters $\lambda=0.2$ and $\alpha=\beta = 3$. Truncation order: $2$ (in top),  $6$ (in middle),  $100$ (in bottom).}
\label{Fig:gamma}
\end{figure}

\subsection{Bivariate beta Lancaster copulas}

Consider $\di (X,Y)$ following the triangular bivariate beta distribution (see \cite{Buja1990})
$$\di \sigma(dx,dy)=\frac{a+b}{\mathbf{B}(a,b)}x^{a-1}y^{b-1}\mathds{1}_{\mathcal{A}}(x,y), $$
where  $$\mathcal{A}=\{(x,y): x,y\geq 0, x+y \leq 1, 1\leq a,b < \infty\}$$
and $\di \mathbf{B}$ denotes the beta function. 
The marginal distribution of $X$ is $\di B_{a,b+1}, \text{ i.e., }\mu(dx)\propto B_{a,b+1}(dx)$, where $B_{a,b+1} $ denotes the univariate beta distribution with parameters $a$ and $b+1$. This distribution has the shifted Jacobi 
polynomials, with parameter $\alpha=a+1$ and $\beta=b$ on the interval $\di (0,1)$, as its orthonormal polynomials. Similarly, $\di \nu(dy)\propto B_{b,a+1}(dy)$. The Lancaster sequence is given by $$\di \rho_{n}=\frac{(-1)^{n}\sqrt{ab}}{\sqrt{(a+n)(b+n)}}.$$
As a result, the copula density function and the copula function are expressed as a Jacobi series that involves
the marginal quantile transforms of the beta distributions. 

\subsection{Construction of Lancaster copulas}
The class of Quadratic Natural Exponential Families (QNEFs),  
 characterized in \cite{Morris}, includes  classical distributions such as Gaussian, Poisson, binomial, negative binomial, gamma, and hyperbolic distributions. These  distributions have the particularity that their variance can be expressed as a quadratic form of their mean. 
Lancaster  \cite{Lancaster1975} proved the following result:  
\begin{theorem}\label{theoLancaster}[Lancaster, 1975] Let $U, V$  be two independent random variables in a QNEF  and let $X = U  $ and $Y = U + V$. Then the joint distribution of $(X, Y)$ is a Lancaster probability.
\end{theorem}
As a direct consequence of Theorem \ref{theoLancaster},  we obtain a simple construction of  Lancaster copulas associated with QNEFs as follows:  
\begin{corollary}\label{cor1}
  Let $U$ and $V$ be two continuous random variables with distribution $\mu$  in a QNEF. Let $(Q_n)_{n\in\N}$ be a basis of $\mu$-orthonormal polynomials. 
  Then the joint distribution $\si$ associated with  $(U,U+V)$ has a Lancaster copula density given by 
  \begin{align*}
\di c(u,v)&= \sum_{n \in \N}\rho_{n}Q_{n}(F_{U}^{-1}(u))\widetilde Q_{n}(F_{U+V} ^{-1}(v)), 
  \end{align*}
  where 
  \begin{align*}
      \rho_n & = 
     \E(Q_n(U)Q_n(U+V)), 
  \end{align*}
  and $(\widetilde Q_n)_{n\in\N}$ is the basis of orthonormal polynomials associated with  the 2-convolution measure $\mu^{*2}$.  
\end{corollary}

\section{Dependence measures}
\label{sectail}
In this section, we focus on two classical ways of measuring dependence: rank-based measures, which capture monotone associations and tail dependence, which describes the behavior of extreme co-movements.
\subsection{Spearman's rho and Kendall's tau} 

Let $(X,Y)$ be a continuous bivariate random variable with copula $C$. Two widely used rank-based measures are Spearman's $\rho$ and Kendall's $\tau$ given by 
(see Schweitzer and Wolff \cite{Schweitzer}) %
\begin{eqnarray*}
\rho_C = 12 \int_0^1 \int_0^1 C(u,v) \, du \, dv - 3, \quad
\tau_C = 4 \int_0^1 \int_0^1 C(u,v)\, dC(u,v) - 1.
\end{eqnarray*}
These definitions highlight that both coefficients depend only on the dependence structure
encoded by the copula and not on the marginal distributions themselves. In what follows, we derive explicit expressions for $\rho_C$ and $\tau_C$ in terms of the copula coefficients $\rho_n$ when $C$ is a Lancaster copula. 

We begin with a technical lemma on integrals of shifted Legendre polynomials, which will be used repeatedly in the proofs below.
\begin{lem}\label{lem:Legendre}
Let $(L_n)_{n\ge0}$ the shifted Legendre polynomials on $[0,1]$, orthonormal in $L^2(0,1)$. Then we have: 
  \beg\label{lem:Legen_eq1}
  \int_{0}^{1}u L_{n}(u) du&=& \frac{1}{2}  \delta_{n,0}+  \frac{\sqrt{3}}{6}\delta_{n,1}, \! \forall n\geq 0
\\
\label{lem:Legen_eq2}
\di \int_{0}^{1}\Bigg(\int_{0}^{u}L_{n}(s)ds\Bigg)du &=&  \frac{1}{2}  \delta_{n,0} - \frac{\sqrt{3}}{6}\delta_{n,1}, \! \forall n\geq 1 
\\
\label{lem:Legen_eq3}
\di\int_0^u L_n(s)\,ds &=&\alpha_n L_{n+1}(u) - \beta_n L_{n-1}(u)-\alpha_n L_{n+1}(0) + \beta_n L_{n-1}(0), \! \forall n\geq 1,
\en
where
\be*
\alpha_n = \frac{\sqrt{2n+1}}{2(2n+1)\sqrt{2n+3}} & {\rm and } & 
\beta_n = \frac{\sqrt{2n+1}}{2(2n+1)\sqrt{2n-1}}.
\e*
\end{lem}
\begin{proof} 
%
(\ref{lem:Legen_eq1}) follows from the basic following properties: 
 $L_0(u)=1$, $L_1(u)=\sqrt{3}(2u-1)$ and $\di \int_0^1 L_n(u)L_k(u)\,du=\delta_{nk}$ for all $n, k \geq 0.$ 
%
By integration by parts,
\be*
\int_0^1\left(\int_0^u L_n(s)\,ds\right)du
=\int_0^1 L_n(u)\,du - \int_0^1 uL_n(u)\,du.
\e*
and (\ref{lem:Legen_eq2}) follows.
For $n\ge1$, the shifted orthonormal Legendre polynomials satisfy the derivative identity (see \cite{Szego}) 
$$
\di \frac{2(2n+1)}{\sqrt{2n+1}}\,L_n(u) =
\frac{1}{\sqrt{2n+3}}\,L'_{n+1}(u) - \frac{1}{\sqrt{2n-1}}\,L'_{n-1}(u),
\qquad u\in(0,1).
$$
Integrating both sides from $0$ to $u$, we obtain
\[
\begin{split}
2\sqrt{2n+1}\int_0^u L_n(s)\,ds
={}&
\frac{L_{n+1}(u)-L_{n+1}(0)}{\sqrt{2n+3}}
\\
&
-\frac{L_{n-1}(u)-L_{n-1}(0)}{\sqrt{2n-1}}.
\end{split}
\]
which gives (\ref{lem:Legen_eq3}).
%
\end{proof}

\begin{proposition}
Assume that \(C\) is a  Lancaster copula whose standardized density admits the representation (\ref{cop_dens_stand}). 
Then
\[
\rho_C=\theta_{1,1}.
\]
Moreover,
\[
\tau_C
=
\frac{2}{3}\theta_{1,1}
+
4\sum_{n,k\ge 1}\theta_{n,k}
\Bigl(
\alpha_n\alpha_k\theta_{n+1,k+1}
-\alpha_n\beta_k\theta_{n+1,k-1}
-\beta_n\alpha_k\theta_{n-1,k+1}
+\beta_n\beta_k\theta_{n-1,k-1}
\Bigr),
\]
where coefficients involving an index \(0\) are understood to be equal to \(0\) in the last sum, and
\[
\alpha_n=
\frac{\sqrt{2n+1}}
{2(2n+1)\sqrt{2n+3}},
\qquad
\beta_n=
\frac{\sqrt{2n+1}}
{2(2n+1)\sqrt{2n-1}}.
\]
\end{proposition}

\begin{proof}
We have 
\[
c(u,v)
=
1+\sum_{n,k\ge1}\theta_{n,k}L_n(u)L_k(v),
\]
and, after integration,
\[
C(u,v)
=
uv+\sum_{n,k\ge1}\theta_{n,k}I_n(u)I_k(v),
\]
where
\[
I_n(u)=\int_0^u L_n(s)\,ds.
\]

We first consider Spearman's rho. Using
\[
\rho_C=12\int_0^1\int_0^1 C(u,v)\,du\,dv-3,
\]
we obtain
\[
\rho_C
=
12\sum_{n,k\ge1}\theta_{n,k}
\left(\int_0^1 I_n(u)\,du\right)
\left(\int_0^1 I_k(v)\,dv\right).
\]
By Lemma 2,
\[
\int_0^1 I_n(u)\,du
=
-\frac{\sqrt3}{6}\delta_{n,1}.
\]
Therefore only the term \((n,k)=(1,1)\) contributes, and
\[
\rho_C
=
12\,\theta_{1,1}
\left(-\frac{\sqrt3}{6}\right)^2
=
\theta_{1,1}.
\]

We now turn to Kendall's tau:
\[
\tau_C
=
4\int_0^1\int_0^1 C(u,v)c(u,v)\,du\,dv-1.
\]
Using the above expressions of \(C\) and \(c\), we write
\[
\frac{\tau_C+1}{4}
=
T_1+T_2+T_3+T_4,
\]
where
\[
T_1=
\int_0^1\int_0^1 uv\,du\,dv
=
\frac14,
\]
\[
T_2=
\sum_{p,q\ge1}\theta_{p,q}
\left(\int_0^1 uL_p(u)\,du\right)
\left(\int_0^1 vL_q(v)\,dv\right),
\]
\[
T_3=
\sum_{n,k\ge1}\theta_{n,k}
\left(\int_0^1 I_n(u)\,du\right)
\left(\int_0^1 I_k(v)\,dv\right),
\]
and
\[
T_4=
\sum_{n,k,p,q\ge1}
\theta_{n,k}\theta_{p,q}
\left(\int_0^1 L_p(u)I_n(u)\,du\right)
\left(\int_0^1 L_q(v)I_k(v)\,dv\right).
\]

By Lemma 2,
\[
\int_0^1 uL_p(u)\,du
=
\frac{\sqrt3}{6}\delta_{p,1},
\]
and
\[
\int_0^1 I_n(u)\,du
=
-\frac{\sqrt3}{6}\delta_{n,1}.
\]
Hence
\[
T_2=T_3=\frac{1}{12}\theta_{1,1}.
\]

It remains to compute \(T_4\). From Lemma 2,
\[
I_n(u)
=
\alpha_nL_{n+1}(u)-\beta_nL_{n-1}(u)
-\alpha_nL_{n+1}(0)+\beta_nL_{n-1}(0).
\]
Since \(p\ge1\), the constant terms vanish after integration against \(L_p\), and therefore
\[
\int_0^1 L_p(u)I_n(u)\,du
=
\alpha_n\delta_{p,n+1}
-
\beta_n\delta_{p,n-1}.
\]
Similarly,
\[
\int_0^1 L_q(v)I_k(v)\,dv
=
\alpha_k\delta_{q,k+1}
-
\beta_k\delta_{q,k-1}.
\]
Injecting these identities into \(T_4\) gives
\[
T_4
=
\sum_{n,k\ge1}\theta_{n,k}
\Bigl(
\alpha_n\alpha_k\theta_{n+1,k+1}
-\alpha_n\beta_k\theta_{n+1,k-1}
-\beta_n\alpha_k\theta_{n-1,k+1}
+\beta_n\beta_k\theta_{n-1,k-1}
\Bigr),
\]
where coefficients with an index \(0\) are taken equal to \(0\) in this expression.

Combining the four terms,
\[
\tau_C
=
4\left(
\frac14+\frac{1}{12}\theta_{1,1}
+\frac{1}{12}\theta_{1,1}
+T_4
\right)-1,
\]
we obtain 
\[
\tau_C
=
\frac{2}{3}\theta_{1,1}
+
4T_4.
\]
This proves the result.
\end{proof}

\subsection{Tail dependence}
Tail dependence is a measure of strength of dependence in the joint lower or joint upper tail of a multivariate distribution. 
In the bivariate case, the concept of tail dependence consists of the amount of dependence in the upper-quadrant tail 
or the lower-quadrant tail of a bivariate distribution (see, e.g., Joe \cite{Joe2014} for  details).
The upper tail dependence coefficient  of a pair of random variables $(X,Y)$ associated with a bivariate copula
$C(u,v)$  is defined by
\be*
\lambda_{U}= \lim_{u\to 1^{-}} \P\Big(Y>F_{Y}^{-1}(u)|X>F_{X}^{-1}(u)\Big)=\lim_{u\to 1^{-}}\frac{ 1-2u+C(u,u)}{1-u}.  
\e*
Analogously, the lower tail dependence coefficient  of $(X,Y)$ is 
\be*
\lambda_{L}=\lim_{u\to 0^{+}}\P\Big(Y\leq F_{Y}^{-1}(u)|X\leq F_{X}^{-1}(u)\Big)=\lim_{u\to 0^{+}}\frac{C(u,u)}{u}. 
\e*
Note that $\lambda_{U}$ and $\lambda_{L}$ are exclusively determined by $C(u,v)$ and do not depend on the marginal distributions. 

\begin{proposition}
Assume that Assumptions (A1)--(A2) hold. Let \(C\) be a Lancaster copula such that its standardized coefficients satisfy
\[
\sum_{n,k\ge1}
|\theta_{n,k}|
\sqrt{(2n+1)(2k+1)}
<\infty.
\tag{20}
\]
Then the upper and lower tail dependence coefficients satisfy
\[
\lambda_U=\lambda_L=0.
\]
\end{proposition}

\begin{proof}
We first show that \(c\) is bounded on \([0,1]^2\). Recall that the orthonormal shifted Legendre polynomials satisfy
\[
\|L_n\|_\infty^2\le 2n+1,
\qquad n\ge0.
\]
Hence,
\[
|L_n(u)L_k(v)|
\le
\sqrt{(2n+1)(2k+1)},
\qquad (u,v)\in[0,1]^2.
\]
Using the standardized expansion
\[
c(u,v)
=
1+\sum_{n,k\ge1}\theta_{n,k}L_n(u)L_k(v),
\]
we obtain
\[
|c(u,v)|
\le
1+
\sum_{n,k\ge1}
|\theta_{n,k}|
\sqrt{(2n+1)(2k+1)}.
\]
Assumption (20) therefore implies that \(c\) is bounded on \([0,1]^2\). More precisely,
\[
\|c\|_\infty
\le
1+
\sum_{n,k\ge1}
|\theta_{n,k}|
\sqrt{(2n+1)(2k+1)}
=:M<\infty.
\]

We now study the lower tail dependence coefficient. Since
\[
C(u,u)
=
\int_0^u\int_0^u c(s,t)\,ds\,dt,
\]
we have
\[
0\le C(u,u)
\le
M u^2.
\]
Therefore,
\[
0\le
\frac{C(u,u)}{u}
\le
Mu
\underset{u\to0^+}{\longrightarrow}0.
\]
Hence
\[
\lambda_L=0.
\]

For the upper tail dependence coefficient, we note that
\[
1-2u+C(u,u)
=
\mathbb P(U>u,V>u)
=
\int_u^1\int_u^1 c(s,t)\,ds\,dt.
\]
Thus,
\[
0\le
1-2u+C(u,u)
\le
M(1-u)^2.
\]
Dividing by \(1-u\), we obtain
\[
0\le
\frac{1-2u+C(u,u)}{1-u}
\le
M(1-u)
\underset{u\to1^-}{\longrightarrow}0.
\]
Consequently,
\[
\lambda_U=0.
\]
\end{proof}


\subsection{Interpretation} 
The results on tail dependence  mean that Lancaster copulas do not exhibit asymptotic dependence in the tails of the distribution. As for Gaussian or Frank copulas, Lancaster copulas model moderate or weak dependencies but fail to capture extreme co-movements. 
It implies that Lancaster  margins are not asymptotically dependent, i.e., if one variable takes an extreme value, the probability that the other variable also takes an extreme value tends to zero.


This limitation is shared by other widely used copula families based on smooth densities,
such as Gaussian or Frank copulas.
Lancaster copulas are therefore primarily intended for modeling moderate or weak dependence
structures, where tail asymptotic independence is acceptable, for instance in reliability,
biometrics, or dependence modeling driven by central rather than extreme behavior.



Concerning Spearman and Kendall coefficients, the 
 Spearman's rho depends exclusively on the first-order
coefficient and is therefore insensitive to higher-order components of the
copula expansion. Kendall's tau incorporates additional interaction terms
involving successive coefficients $\theta_{n\pm 1,k\pm 1}$. These terms are
multiplied by weights that decay quadratically with $n$ and $k$, implying that
higher-order contributions have a limited impact provided the sequence 
$(\theta_{n,k})$ decays sufficiently fast. This shows that in the standardized Lancaster family, the first coefficient \(\theta_{1,1}\) fully controls the dominant rank correlation, while the higher-order coefficients shape finer features of the copula without significantly affecting \(\rho_C\) or \(\tau_C\).

\section{Multivariate extensions}\label{secmulti}

To extend the definition of the Lancaster copula to the multivariate case, we introduce the following multivariate notation. Given $x = (x_1,\cdots,x_d)\in \R^d$ and 
$n = (n_1,\cdots,n_d)\in \N^d$, we write $x^n:=x_{1}^{n_1}\times \cdots \times x_{d}^{n_d}$  and we call the integer
$\vert n\vert := n_1+\cdots+n_d$ the order of $n$. A $\vert n\vert$-th degree polynomial, say $P_n$, has
the following form
\be* 
\di P_n(x)=\sum_{\substack{k\in \N^d\\ \vert k\vert\leq \vert n\vert}}\alpha_k x^k,
\e*
where at least one coefficient $\alpha_k$ is nonzero for
$\vert k\vert = \vert n\vert$. 
Let 
${X}$ and $Y$ be two random variables on $\R^d$ 
with probability distributions $\mu$ and $\nu$, respectively.
Let $(P_n)_{n\in \N^d}$ and $(Q_n)_{n\in \N^d}$ be two bases of orthonormal polynomials with respect to the distributions of ${X}$ and $Y$, respectively. 
Here and subsequently, $P_k$ and $Q_n$ denote polynomials of the $\vert k\vert$th and $\vert n\vert$th degrees.

The joint distribution $\sigma$ is called a Lancaster
probability if 
\be* 
\di \E\left(P_k({X})Q_n(Y)\right) =\rho_n \delta_{k,n}. 
\e*
The sequence $(\rho_n)_{n\in\N^d}$  is called a Lancaster sequence and we assume that 
$\sum_{n\in \N^d} \rho_n^2 < \infty$. 
Then we can write
\beg\label{poly_lancaster}
\di \sigma(dx,dy) = \sum_{n\in \N^d}\rho_n P_n(x)Q_n(y)\mu(dx)\nu(dy).
\en  
We assume that $\mu$, $\nu$ and $\sigma$  are
absolutely continuous with respect to the Lebesgue measure and we denote their corresponding probability density functions by $f_{X}$,$f_{Y}$ and $f_{X,Y}$ respectively.
This yields the expansion
\beg \label{dens_lancaster2}
\di f_{{X},Y}(x,y)=f_{{X}}(x)f_{Y}(y) \sum_{n\in \N^d}\rho_n P_n(x)Q_n(y).
\en 

Assume that $f_X>0$ and $f_Y>0$, identifying (\ref{cop1_chap5}) and (\ref{dens_lancaster2}), we obtain
the relation
\be*
\label{cop0_chap5}
c( u,  v) & = & \di \sum_{ n\in\N^d}
\rho_{ n}
P_n\left(F_{{X}}^{-1}(u)\right)Q_n\left(F_{Y}^{-1}( v)\right),
\e*
where $ u=(u_1,\cdots,u_d),  v=(v_1,\cdots,v_d) \in [0,1]^d$, 
$ F_{X}^{-1}(x)= (F_{X_1}^{-1}(x_1),\cdots,F_{X_d}^{-1}( x_d))$, and $ F_{Y}^{-1}(y)= (F_{Y_1}^{-1}(y_1),\cdots,F_{Y_d}^{-1}( y_d))$.

We deduce an expression of the copula $C$ associated with $(X, Y)$, as follows
\be*
C( u,  v) & = & \di \sum_{n\in\N^d}
\rho_{n}\left(\int_{-\infty}^{F_{X_1}^{-1}(u_1)}\cdots
\int_{-\infty}^{F_{X_d}^{-1}(u_d)}P_n(x)\mu(dx)\right)\left(\int_{-\infty}^{F_{Y_1}^{-1}(v_1)}\cdots\int_{-\infty}^{F_{Y_d}^{-1}(v_d)} Q_n(y)\nu(dy)\right).
\e*

\subsection{Multivariate tail dependence}   
We  define multivariate lower tail dependence measures as suggested in \cite{Joe2014} by:
\be* 
\di \lambda_{L}(C)=\lim_{u \to 0^{+}}\frac{C(u\mathbf{1}_{d})}{u}, \  ~\text{with } u\mathbf{1}_{d}:=(u,\cdots,u).
\e*

Write (A1')-(A2') the multivariate versions of (A1)-(A2).   
\begin{prop}\label{upper tail multi}
Assume that (A1')-(A2') hold. Let $C$ be a  multivariate Lancaster copula such that 
\be*
\sum_{n,k \in\mathbb{N}^d\setminus\{0\}}
|\theta_{ n,k}|\,
\prod_{j=1}^d \sqrt{2n_j+1}\sqrt{2k_j+1}
<\infty ,
\e*
where $\theta_{n,k}$ are the coefficients associated to the multivariate standardized copula. 
Then the multivariate lower tail dependence measure satisfies
\be*
\di \lambda_{L}(C)=0.
\e*
\end{prop}

\begin{proof}
For any $ u\in[0,1]^d$,
\[
\left|\prod_{j=1}^d L_{n_j}(u_j)\right|
\le \prod_{j=1}^d \|L_{n_j}\|_\infty
\le \prod_{j=1}^d \sqrt{2n_j+1},
\]
and by the triangle inequality we get
\[
|c( u)|
\le 1+\sum_{ n,k\neq \bm 0}|\theta_{ n,k}|
\prod_{j=1}^d \|L_{n_j}\|_\infty \|L_{k_j}\|_\infty
\le 1+\sum_{ n, k \neq \bm 0}|\theta_{ n,k}|
\prod_{j=1}^d \sqrt{2n_j+1}\sqrt{2k_j+1}
<\infty,
\]
hence $c\in L^\infty([0,1]^d)$; denote $M:=\|c\|_\infty<\infty$.

Since $C$ has density $c$,
\[
C(u\bm 1_d)=\int_{[0,u]^d} c( s)\,d s.
\]
Thus, for $u\in(0,1)$,
\[
0\le C(u\bm 1_d)\le \int_{[0,u]^d} \|c\|_\infty\,d s
=M\,u^d.
\]
Dividing by $u>0$ gives
\[
0\le \frac{C(u\bm 1_d)}{u}\le M\,u^{d-1}\xrightarrow[u\to 0+]{}0
\qquad\text{(since $d\ge 2$),}
\]
which proves $\lambda_L(C)=0$. 
\end{proof}


\subsection{Gaussian illustration} 
Assume that $X=(X_1,X_2)$  and  $ X'=(X_1',X_2')$ are  two independent  Gaussian vectors with mean $m$ and variance $\Sigma$ (resp. $m'$ and $\Sigma'$). Write 
$ X'' =  X +  X'$. Then from the multidimensional version of Theorem \ref{theoLancaster} (see Theorem 3 in  \cite{Koudou2000})  the distribution of $( X,  X'')$ is Lancaster with density: 
\begin{align*}
    f_{ X,  X''}(( x ,  x'')) & = 
    f_{ X}( x) f_{ X''}(x'') 
    \di\sum_{(n_1,n_2)\in \N^2} 
    \rho_{n_1,n_2} 
    H_{n_1}(x_1)H_{n_2}(x_2)H_{n_1}(x''_1)H_{n_2}(x''_2), 
\end{align*}
with 
\begin{align*}
    \rho_{n_1,n_2} &= 
    \E( H_{n_1}(X_1)H_{n_2}(X_2)H_{n_1}(X''_1)H_{n_2}(X''_2)),
\end{align*}
where $H_n$ are orthonormal Hermite polynomials. In the case where $\Sigma$ (resp. $\Sigma'$) is diagonal, that is, when $X_1$ and $X_2$ (resp. $X_1'$ and $X_2'$)  are independent, we obtain directly: 
\begin{align*}
    \rho_{n_1,n_2} &= 
    \E( H_{n_1}(X_1)H_{n_1}(X_1''))\E(H_{n_2}(X_2)H_{n_2}(X''_2))
    \\
    & = 
    r_1^{n_1} r_2^{n_2}, 
\end{align*}
where $r_i={\rm  cor}(X_i,X_i'')$, $i=1,2$. In the particular case where 
$$\Sigma = \Sigma' = \left( \begin{array}{cc} 1 & 0 
\\ 0 & 1
\end{array}
\right) $$  we have $r_1=r_2=1/\sqrt{2}$  
and we deduce 
\be*
\label{cop0_chap5}
c(u, v) & = & \di \sum_{ (n_1,n_2)\in \N^2}
\alpha_{n_1,n_2}
H_{ n_1}(F_{X_1}^{-1}(u_1))H_{ n_1}(F_{X''_1}^{-1}(v_1))H_{ n_2}(F_{X_2}^{-1}(u_2))H_{ n_2}(F_{X_2''}^{-1}(v_2)),
\e*
with $\alpha_{n_1,n_2}=(2^{(n1+n2)/2})^{-1}$.  
\section{Discussion}
\label{discussion} 
We introduce \textit{Lancaster copulas}, a new class of copulas built from bi-orthogonal expansions of Lancaster probabilities. Under assumptions (A1)-(A2) satisfied by all our illustrations, we derive explicit infinite-series representations for copulas and copula densities and study practical truncations. Numerically, low-order truncations already provide accurate approximations across several canonical families (exponential/Downton, Gaussian, gamma, and beta), which we attribute to the fast decay of the expansion coefficients, and  to non-extreme tails that do not require high expansion coefficients $\rho_n$. We further analyze tail behavior and show that Lancaster copulas exhibit no tail dependence. An extension to the multivariate setting is outlined. 

Through the invariance property, Lancaster copulas enable us to understand the dependence structure of a new family of distributions, which we call {\it generalized Lancaster distributions}, that generalize the class of continuous Lancaster distributions by replacing the bi-orthogonal polynomials with families of bi-orthogonal functions, while preserving the same copula. 

Finally, the good approximation results demonstrated numerically pave the way for the estimation of Lancaster copulas. 
More precisely, if we  observe  continuous iid observations
$(X_1,Y_1),\cdots, (X_m,Y_m)$,  from a  Lancaster distribution with margins $\mu$ and $\nu$, we can estimate the coefficient  $\rho_{n}$ by
\be*
\di \w{\rho}_{n} & = & \frac{1}{m}\di\sum_{i=1}^m P_{n}(X_i)Q_{n}(Y_i).
\e*
Therefore, a $N$-th order non-parametric estimator of the Lancaster copula density $c$ is given by
\be*
\di \w{c}^{[ N]}(u,v) = \sum_{n=0}^{N}\w \rho_{n}P_{n}(\w F_{X}^{-1}(u))Q_{n}(\w F_{Y}^{-1}(v)),
\e*
where $\w F^{-1}$ is the empirical quantile estimator.
By integration, we get a non-parametric  estimator of the Lancaster copula function  as 
\be*
\di \w{C}^{[N]}(u,v) =\sum_{n=0}^{N}\w\rho_{n}\int_{-\infty}^{\w F_{X}^{-1}(u)}P_{n}(x)\mu(dx)\int_{-\infty}^{\w F_{Y}^{-1}(v)}Q_{n}(y)\nu(dy).
\e*
We propose to use a data-driven procedure based on the Least-Squares
Cross-Validation (LSCV) to select the optimal truncation parameter
$\di \w{N}_{opt}$. We determine $\di \w{N}_{opt}$ by minimizing the 
LSCV criterion, a cross-validation technique introduced by \cite{rudemo} 
and \cite{bowman} for kernel density estimation, and subsequently adapted 
to orthogonal series expansions by \cite{taylor, Ngounou2023}. 
The criterion is given by
\be* \label{cross-validation}
\di LSCV(N)= \int_{I^{2}} \left(\w{c}^{[ N]}(u,v)\right)^2 dudv -\di \frac{2}{m}\di \sum_{i=1}^{m}\w{c}^{[N]}_{-i}\left(F_{X}(X_{i}), F_{Y}(Y_{i})\right), 
\e*
with its empirical counterpart
\be* \label{cross-validation-empirical}
\di \w{LSCV}(N)= \int_{I^{2}} \left(\w{c}^{[ N]}(u,v)\right)^2 dudv -\di \frac{2}{m}\di \sum_{i=1}^{m}\w{c}^{[N]}_{-i}\left(\w{F}_{X}(X_{i}), \w{F}_{Y}(Y_{i})\right), 
\e* 
defining the estimator
\be*
\di \w{N}_{opt} =\argmin_{N\in \N} \w{LSCV}(N),
\e* 
where  $\di \w{c}^{[ N]}_{-i}$ denotes the Lancaster estimator computed
without observation  $(X_{i},Y_{i})$.

This would allow for the estimation of Lancaster copulas and their densities, and for the derivation of asymptotic properties, similar to the work conducted in \cite{Ngounou2023}.
Moreover, it suggests the possibility of a test to determine whether a vector admits a Lancaster copula.

\section*{Acknowledgments}
D. Pommeret is a member of the ANR DREAMES project. He also gratefully acknowledges the support of the ACTIONS Chair, under the aegis of BNP Paribas Cardif, in collaboration with the Institut des Actuaires and the Fondation du Risque.

\end{document}